\documentclass[11pt]{article}
\usepackage{amsfonts,amsmath,amsthm,amssymb,graphicx,mathrsfs}

\usepackage[top=2cm,
			bottom=2cm,
			left=2cm,
			right=2cm]{geometry}
\usepackage[colorlinks=true,urlcolor=blue, linkcolor =blue,citecolor=blue]{hyperref}
\usepackage{algorithm,algorithmicx,algpseudocode,multirow}
\usepackage{booktabs}

\theoremstyle{plain}
\newtheorem{theorem}{Theorem}
\newtheorem{dfn}{Definition}
\newtheorem{lemma}{Lemma}
\newtheorem{assume}{Assumption}
\newtheorem{remark}{Remark}
\newtheorem{corollary}{Corollary}

\usepackage{breakcites}

\allowdisplaybreaks

\begin{document}

\title{\bf Optimal Quasi-Bayesian reduced rank regression with incomplete response}

\author{The Tien Mai\footnote{Corresponding author, Email: the.t.mai@ntnu.no} $\,^{ , (1)} $ \& Pierre Alquier$^{(2)} $ }

\date{
\begin{small}
$^{(1)} $Department of Mathematical Sciences, 
\\
Norwegian University of Science and Technology, Norway.
\\
$^{(2)} $ RIKEN AIP, Japan.
\end{small}
}

\maketitle

\begin{abstract}
The aim of reduced rank regression is to connect multiple response variables to multiple predictors. This model is very popular, especially in biostatistics where multiple measurements on individuals can be re-used to predict multiple outputs. Unfortunately, there are often missing data in such datasets, making it difficult to use standard estimation tools. In this paper, we study the problem of  reduced rank regression where the response matrix is incomplete. We propose a quasi-Bayesian approach to this problem, in the sense that the likelihood is replaced by a quasi-likelihood. We provide a tight oracle inequality, proving that our method is adaptive to the rank of the coefficient matrix. We describe a Langevin Monte Carlo algorithm for the computation of the posterior mean. Numerical comparison on synthetic and real data show that our method are competitive to the state-of-the-art where the rank is chosen by cross validation, and sometimes lead to an improvement.
\end{abstract}


\section{Introduction}
Reduced rank regression (RRR) is a popular model widely used in studying the relationship between multiple response variables and a set of predictors \cite{anderson1951estimating,izenman1975reduced,izenman2008modern,velu2013multivariate}. Connecting $m$ predictors to $p$ response variables via a linear relation requires the estimation of $mp$ coefficients, which might be difficult or impossible when $p$ and/or $m$ is large. RRR uses a low-rank constraint to reduce the dimension of the problem. This dimension reduction makes accurate estimation possible in large dimension \cite{bunea2011optimal}. Moreover, the span of the regression matrix can receive a nice interpretation in terms of a small number of latent variables explaining all the response variables. Various methods were proposed for estimation in RRR, including frequentist and Bayesian approaches \cite{geweke1996bayesian,kleibergen2002priors,corander2004bayesian,schmidli2019bayesian,alquier2013bayesian}. It was also proposed to add more constraints in the model, such as sparsity \cite{goh2017bayesian,chakraborty2020bayesian,yang2020fully}.

Bayesian RRR has been successfully applied in many applications, for example in biostatistics and genomics \cite{marttinen2014assessing,zhu2014bayesian}. However, in this field, most datasets have missing data \cite{lee2002transcriptional,marttinen2014assessing} and data with missing value are often removed or imputed before applying the model. However, incorrect imputation might result in estimation bias, and we are not aware of any theoretical guarantees for these procedures in this context. From a theoretical perspective, RRR is a special case of the so-called trace regression model. In \cite{koltchinskii2011nuclear}, the authors study penalized least-square estimation in the trace regression model, and prove tight oracle inequalities. The reference \cite{luo2018leveraging} studied some penalized likelihood estimation procedure for RRR with incomplete response. However, there is no current available Bayesian approach could deal with missing data in RRR.

In this paper, we focus on quasi-Bayesian estimation in the reduced rank regression model with missing observations, in the sense that the response matrix is incomplete. We propose to use a multivariate Student prior. We prove that the posterior mean satisfies a tight oracle inequality: even when the true rank is unknown, it converges at the optimal rate we could hope when knowing the rank. We also prove a contraction of the posterior result. We develop a Langevin Monte Carlo (LMC) approach to compute the posterior mean and to sample from the posterior. We show that it is comparable to the frequentist estimator on both simulated and real datasets.

Quasi-Bayesian estimation is an extension of the Bayesian approach where the quality of the data fit is not necessarily measured by the likelihood, but by a more general notion of risk or a quasi-likelihood. This approach is increasingly popular in generalized Bayesian inference and machine learning \cite{bissiri2013general,kno2019}. First, it allows to avoid restrictive assumptions on the data generating process. Moreover, it allows to focus on some aspects of the problem in mind (for example, prediction rather than estimation).

Our theoretical results on quasi-posteriors are derived from PAC-Bayes bounds. These bounds were introduced by \cite{McA,see2002,lan2002,mau2004,ger2009} to provide numerical generalization certificates on quasi-Bayesian estimators. They were later extended by \cite{catoni2004statistical,catonibook} as tools to provide oracle inequalities for such estimators, this approach is strongly related to the so-called ``information bounds'' of \cite{zha2006,russo2019much}. We refer the reader to \cite{gue2019,TUTO} for introductions to this topic. PAC-Bayes bounds were used to prove oracle inequalities for matrix estimation problems such as matrix completion \cite{mai2015,cottet20181} and quantum tomography \cite{mai2017pseudo}. One of the byproducts of this approach is that we don't have to assume anything about the distribution of the missing entries in the response matrix. This is in contrast with previous works on matrix completion such as \cite{candes2010matrix, koltchinskii2011nuclear} where the location of the missing entries is assumed to be uniform. These assumptions were relaxed in further works \cite{foygel2011learning,klopp2014noisy,
negahban2012restricted}, but we are not aware of any result without any assumption on this distribution.

The idea that multivariate Student priors lead to optimal rates in high-dimensional estimation problems is due to \cite{dalalyan2008aggregation,dalalyan2012sparse}. Since then, these priors were used in matrix completion \cite{yang2018fast,mai2022} and image denoising \cite{dalalyan2020exponential}. Although the scaled spectral Student prior is not conjugate in our problem, it is particularly convenient for implementing gradient-based sampling method. We propose an LMC algorithm for sampling from the (quasi) posterior and for the computation of the posterior mean. The LMC method was introduced in physics based on Langevin diffusions \cite{ermak1975computer} and became popular in statistics and machine learning following the paper \cite{roberts1996exponential}. Recent advances in the study of LMC make it particularly suitable for high-dimensional problems \cite{dalalyan2017theoretical,durmus2017nonasymptotic,durmus2019high,dalalyan2020sampling}.

The paper is organized as follows. In Section \ref{section:notations} we introduce the notations for the RRR model, the quasi-Bayesian approach and our spectral Student prior. In Section \ref{sc_theory} we prove the convergence of the posterior mean, and the contraction of the posterior. In Section \ref{sc_LMC} we describe the Langevin MC method we implemented, we then compare it to the frequentist (penalized) estimator in Section \ref{sc_simu}. All the proofs are gathered in the appendix \ref{sc_appendix_proof}.

\section{Reduced rank regression with incomplete response}
\label{section:notations}

\subsection{Notations}
The set of $n_1 \times n_2$ matrices with real coefficients is denoted by $\mathbb{R}^{n_1\times n_2}$. For any $A\in \mathbb{R}^{n_1\times n_2}$ and $I=(i,j) \in\{1,\dots,n_1\} \times \{1,\dots,n_2\}$, we denote by $A_I=A_{(i,j)}=A_{i,j}$ the coefficient on the $i$-th row and $j$-th column of $A$. We let $A^\intercal \in \mathbb{R}^{n_2\times n_1}$ denote the transpose of $A$. The matrix in $\mathbb{R}^{n_1\times n_2}$ with all entries equal to $0$ is denoted by $\mathbf{0}_{n_1 \times n_2}$. For a square matrix $B\in\mathbb{R}^{n_1 \times n_1}$ we let ${\rm Tr}(B)$ denote its trace. We denote the identity matrix in $ \mathbb{R}^{n_1 \times n_1} $ by $\mathbf{I}_{n_1}$. For $A\in \mathbb{R}^{n_1\times n_2}$, we define its sup-norm $\| A \|_{\infty} = \max_{i,j}|A_{i,j}| $; its Frobenius norm $ \|A\|_F$ is defined by $ \|A\|_F^2 = {\rm Tr}( A^\intercal A) = \sum_{i,j} A_{i,j}^2$ and ${\rm rank}(A)$ its rank. For a probability distribution $P$ on $\{1,\dots,n_1\} \times\{1,\dots,n_2\}$, we generalize this notation by $ \|A\|_{F,P}^2 = \sum_{i,j} P[(i,j)] A_{i,j}^2 $; note that when $P$ is the uniform distribution, then $ \|A\|_{F,P}^2 = \|A\|_{F}^2 /(n_1 n_2) $.

\subsection{Model}
Let $\ell$, $m$ and $p$ be integers: $\ell$ is the number of individuals, $m$ the number of explanatory variables and $p$ the number of response variables. We assume that we observe a design matrix $X \in \mathbb{R}^{\ell \times m}$ and $n$ i.i.d random pairs $ (\mathcal{I}_1, Y_1),
\ldots, (\mathcal{I}_n, Y_n) $ given by
\begin{equation}
\label{main_model}
Y_i = (XM^*)_{\mathcal{I}_i} + \mathcal{E}_i, \quad i = 1, \ldots, n
\end{equation}
where $ M^* \in \mathbb{R}^{m\times p} $ is the unknown matrix to be estimated.
The noise variables $ \mathcal{E}_i $ are assumed to be independent with
$ \mathbb{E} (\mathcal{E}_i) = 0. $ The variables $ \mathcal{I}_i $ are i.i.d copies of a random variable $ \mathcal{I} $ having distribution $ \Pi $ on the set
$ \lbrace 1, \ldots, \textcolor{red}{\ell} \rbrace \times \lbrace 1, \ldots, p \rbrace $, we put $ \Pi_{x,y} := \Pi(\mathcal{I} = (x,y)) $.
We call reduced rank regression (RRR) with missing entries the model in \eqref{main_model} under the assumption that ${\rm rank}(M^*) < \max(m,p)$.

Let us comment on this model. First, when \textcolor{red}{$\ell=m$} and $X=\mathbf{I}_p$ is the identity matrix, we recover the matrix completion as a special case~\cite{koltchinskii2011nuclear}. However, we will focus here in the case where $X$ contains explanatory variables. Also note that there are two approaches to model the observed entries of $Y$: with, or without replacement. In the matrix completion problem, both were studied, see for example \cite{candes2010matrix} for the case without replacement and \cite{koltchinskii2011nuclear} with replacement. Both settings correspond to practical applications, and the same estimation methods are used in both cases. We chose to develop our theory for i.i.d variables $ \mathcal{I}_i $, which means that it is possible in our context to observe the same entry multiple times. Note that, thanks to the results in Section 6 of~\cite{hoeffding1963probability}, our results can be extended directly to the case of sampling $\mathcal{I}_i$ without replacement on the condition that they are sampled uniformly, and that there is no observation noise: $\mathcal{E}_i=0$. Finally, technical moment assumptions on the $\mathcal{E}_i$ will be provided below. The independence assumption between these variables is standard in regression; it is possible to remove if in matrix completion~\cite{alquier2022tight}, at the cost of technical weak dependence assumptions that we will not discuss here.

Let us now state our assumptions on this model.
\begin{assume}
\label{assum_bounded}
There is a known constant $ C  < +\infty $  such that
$
 \| X M^* \|_{\infty} \leq C.
$
\end{assume}

\begin{assume}
\label{assum_noise}
The noise variables  $ \mathcal{E}_1, \ldots, \mathcal{E}_n$ are independent 
and independent of $\mathcal{I}_1,\ldots,\mathcal{I}_{n}$. There exist two known
constants $\sigma>0$ and $\xi>0$ such that
$$ \mathbb{E} (\mathcal{E}_{i}^{2})\leq \sigma^{2} $$
$$ \forall k\geq 3,\quad \mathbb{E} (|\mathcal{E}_{i}|^{k}) \leq \sigma^{2} k! \xi^{k-2}.$$
\end{assume}
Assumption \ref{assum_bounded} states that $\mathbb{E}(Y_i|\mathcal{I}=(x,y))$ is bounded for any $(x,y)$. In the case of matrix completion, $X=\mathbf{I}_p$ and thus this simply boils down to assuming that $\|M\|_{\infty} \leq C$.  Assumption \ref{assum_noise} states that the noise is sub-exponential: this include  a wide class of possible noises such as  bounded noise and Gaussian noise. We refer the reader e.g. to Chapter 2 in \cite{boucheron2013concentration} for more on sub-exponential variables. Assumption \ref{assum_bounded} and \ref{assum_noise} are both standard, they have been used in \cite{luo2018leveraging} for theoretical analysis of reduced rank regression and in \cite{koltchinskii2011nuclear} for trace regression.

The frequentist methods in this model are based on minimization of the least square criterion with a low-rank inducing penalty. There is a subtlety in our case: under Assumption \ref{assum_bounded}, we know that it makes no sense to return predictions $X M$ with entries that are outside of $[-C,C]$. However, it is extremely convenient to use an unbounded prior for $M$. Thus, we propose to use indeed unbounded distributions for $M$, but to use as a predictor a truncated version of $XM$ rather than $ M $ itself. For a matrix $A$, let 
$$
\Pi_C(A) = \arg\min_{\|B\|_{\infty} \leq C } \|A-B\|_F
$$ be the orthogonal projection of $A$ on matrices with entries bounded by $C$. Note that $B$ is simply obtained by replacing entries of $A$ larger than $C$ by $C$, and entries smaller than $-C$ by $-C$.

For a matrix $ M\in \mathbb{R}^{m\times p} $, we denote by $r(M)$ the ``empirical risk'' or least-square criterion of $M$,
\begin{equation*}
r(M) = \dfrac{1}{n} \sum\limits_{i = 1}^{n} \left( Y_i - (\Pi_C(X M))_{\mathcal{I}_i} \right) ^2.
\end{equation*}
Its expectation is denoted by
\begin{equation*}
R(M) = \mathbb{E} \left[ r(M)  \right]
 = \mathbb{E} \left[  \left( Y_1 - (\Pi_C(XM))_{\mathcal{I}_1} \right)^2  \right] .
\end{equation*}
In this paper, we will focus on the predictive aspects of the model: a matrix $M$ predicts almost as well as $M^*$ if $ R(M) - R(M^*) $ is small. Under the assumption that $\mathcal{E}_i$ has a finite variance, thanks to the Pythagorean theorem, we have
\begin{equation}
 \label{pyta}
 R(M) - R(M^*) =  \| \Pi_C(XM) - XM^*  \|^2_{F,\Pi}
\end{equation}
for any $ M $, which means that our results can also be interpreted in terms of estimation of $M^*$ with respect to a generalized Frobenius norm. 

\subsection{The Quasi-Bayesian approach}

Let $\pi$ be a prior distribution on $\mathbb{R}^{m\times p}$ (we will specify low-rank inducing priors in the Subsection \ref{sc_priors}).
For any $\lambda>0$, we define the quasi-posterior
\begin{equation*}
\hat{\rho}_{\lambda} (dM) \propto \exp(-\lambda
r(M)) \pi (dM).
\end{equation*}
Note that, for $\lambda=n/(2\sigma^2)$, this is exactly the posterior that we would be obtain for a Gaussian noise $\mathcal{E}_i \sim
\mathcal{N}(0,\sigma^2)$ (conditionally on this $\mathcal{I}_i$'s). However, our theoretical results will hold under a more general class of noise. Indeed, it is known that a small enough $\lambda$ will lead to robustness to noise misspecification \cite{grunwald2017inconsistency}. Moreover, even in the case of a Gaussian noise, in high-dimensional settings, taking $\lambda$ smaller than $\lambda=n/(2\sigma^2)$ leads to better adaptation properties \cite{dalalyan2008aggregation,dalalyan2012sparse}. We will actually specify our choice of $\lambda$ below.

The truncated posterior mean of $XM$ is given by
\begin{equation}
\label{equat_estimator}
\hat{XM}_{\lambda} 
= 
\int \Pi_C(XM) \hat{\rho}_{\lambda}(d M).
\end{equation}

\begin{remark}
Let us comment briefly on the projection $\Pi_C$. First, note that using reasonable values for $C$, the Monte Carlo algorithm we used in the simulations never sampled matrices $M$ such that $\Pi_C(X M) \neq XM$. In other words, this projection is necessary for technical reasons, but has very little impact in practice. If one wants to build an estimator of $M^*$ instead of an estimator of $XM^*$, when $X^\intercal X$ is invertible, we can simply define $\hat{M}_\lambda = (X^\intercal X)^{-1} X^\intercal \hat{XM}_{\lambda}$ and note that $X \hat{M}_\lambda =  \hat{XM}_{\lambda}$.
\end{remark}

The quasi-posterior is often referred to as the ``Gibbs posterior'' in the PAC-Bayes approach \cite{catoni2004statistical,catonibook,alquier2015properties,gue2019,TUTO}. For this reason, $\hat{M}_{\lambda}$ is sometimes referred to as the Gibbs estimator, or the exponentially weighted aggregate (EWA) \cite{dalalyan2008aggregation,rigollet2012sparse,dalalyan2018exponentially}.

\subsection{Prior specification}
\label{sc_priors}
We consider the following spectral scaled Student prior, with parameter $\tau>0$,
\begin{align}  
\label{prior_scaled_Student}
\pi(M) 
\propto 
\det (\tau^2 \mathbf{I}_{m} + MM^\intercal )^{-(p+m+2)/2}.
\end{align}
To illustrate that this prior has the potential to encourage the low-rankness of $ M $, one can check that 
\begin{align*}
\pi(M) 
 \propto 
 \prod_{j=1}^{m}  (\tau^2  + s_j(M)^2 )^{- (p+m+2)/2 },
\end{align*}
where $ s_j(M) $ denotes the $j^{th}$ largest singular value of $ M $. It is well known that the log-sum function $ \sum_{j=1}^m \log (\tau ^2  + s_j(M)^2) $ used by \cite{candes2008enhancing,yang2018fast} to enforce approximate sparsity on the singular values $s_j(M)$.  Alternatively, one can recognize a scaled Student distribution evaluated at $ s_j(M)  $ in the last display above which induces approximate sparsity on the $s_j(M)$ \cite{dalalyan2012sparse}. Thus, under the prior, most of the $s_j(M)$ are close to $0$, which means that $M$ is well approximated by a low-rank matrix.

Although this prior is not conjugate in our problem, it is particularly convenient to implement the Langevin Monte Carlo algorithms, see Section \ref{sc_LMC}. This prior has been considered before in the context of image denoising \cite{dalalyan2020exponential}. It can also be seen as the marginal distribution of the Gaussian-inverse Wishart prior that is explored in \cite{yang2018fast} in the context of matrix completion where the precision matrix is integrated out.

\section{Theoretical analysis}
\label{sc_theory}

\subsection{Main results}
In this section, we derive the statistical properties of the posterior $\hat{\rho}_\lambda$ and the mean estimator $\hat{XM}_\lambda$. Let us put 
$
 C_1 = 8(\sigma^2 +C^2) ;\, C_2 = 64C \max(\xi,C);\,
 \tau^* = \sqrt{C_1 (m+p)/(nmp \|X \|_{F}^2)}. 
$

\begin{theorem}
\label{thrm_main} 
Let Assumptions \ref{assum_bounded} and \ref{assum_noise} be satisfied. Fix the parameter $\tau= \tau^* $ in the prior. Fix $\delta>0$ and define $\lambda^* := n \min(1/(2C_2), \delta/[C_1(1+\delta)] ) $. Then, for any $\varepsilon\in(0,1)$, we have, with probability at least $1-\varepsilon$ on the sample,
\begin{multline*}
 \left\|\hat{XM}_{\lambda^*} -XM^* \right\|_{F,\Pi}^2
\leq 
\inf_{0\leq r \leq mp} \inf_{
\begin{tiny}
\begin{array}{c}
\bar{M}\in\mathbb{R}^{p\times m}
\\
{\rm rank}(\bar{M}) \leq r
\end{array}
\end{tiny}
} \Biggl\{ (1+\delta) \|X\bar{M}-X M^*  \|_{F,\Pi}^2
\\
+ 
 \frac{C_1 (1+\delta)^2}{\delta} \frac{ \left( 4 r (m+p+2) \log \left( 1+\frac{\| X \|_F \| \bar{ M } \|_F}{ \sqrt{C_1}} \sqrt{ \frac{nmp}{r(m+p)}} \right) +(m+p) + 2 \log \frac{2}{\varepsilon} \right)}{ n} \Biggr\}.
\end{multline*}
\end{theorem}

Note that the above formula does not actually require ${\rm rank} (\bar{M})\neq 0$, if we have ${\rm rank} (\bar{M})=0$ then $\bar{M}=0$ and we interpret $0\log(1+0/0)$ as $0$. As mentioned in the introduction, the proof of this theorem relies on the PAC-Bayes theory. It is provided in the appendix.
In particular, we can upper bound the infimum on $\bar{M}$ by taking $\bar{M} = M^*$, which leads to the following result.

\begin{corollary}
Under the same assumptions and the same $\tau,\lambda^*$ as in Theorem~\ref{thrm_main}, let $r^*={\rm rank}(M^*)$. Put
$$
R_{\delta,m,p,n,r^*,\varepsilon} : =  \frac{C_1 (1+\delta)^2 }{\delta} \frac{ \left\lbrace 4 r^* (m+p+2) \log \left( 1+\frac{\| X \|_F \| M^* \|_F}{ \sqrt{C_1}} \sqrt{ \frac{nmp}{r^*(m+p)}} \right) +(m+p) + 2 \log \frac{2}{\varepsilon} \right\rbrace}{ n },
$$
 then 
\begin{align*}
 \left\|\hat{XM}_{\lambda^*} -XM^* \right\|_{F,\Pi}^2
\leq 
R_{\delta,m,p,n,r^*,\varepsilon}
\end{align*}
and in particular, if the sampling distribution $\Pi$ is uniform,
\begin{align*}
 \dfrac{   \|  \hat{XM}_{\lambda}   - XM^*  \|^2_{F}   }{ \ell p}
\leq 
R_{\delta,m,p,n,r^*,\varepsilon}.
\end{align*}
\end{corollary}

\begin{remark}
A minimax lower-bound on the estimation of $M^*$ in the trace regression problem (which includes RRR) was derived by \cite{koltchinskii2011nuclear}:    $ r^* \max(m,p) /n $. Thus, $\hat{M}_{\lambda^*}$ is minimax-optimal (at most up to a log term).
\end{remark}

While Theorem~\ref{thrm_main} states that the posterior mean leads to optimal estimation of $XM^*$, it is actually possible to prove that the quasi-posterior $\hat{\rho}_{\lambda}$ contracts around $M^*$ at the optimal rate.
\begin{theorem}
\label{thrm_contraction} 
Under the same assumptions for Theorem~\ref{thrm_main}, and the same definition for $\tau$ and $\lambda^*$, let $\varepsilon_n$ be any sequence in $(0,1)$ such that $\varepsilon_n\rightarrow 0$ when $n\rightarrow\infty$. Define
\begin{multline*}
\mathcal{E}_n 
 = 
 \Biggl\{ M \in \mathbb{R}^{m\times p}: \| \Pi_C(XM) - XM^*  \|^2_{F,\Pi} 
\leq
\inf_{1\leq r \leq mp} \inf_{
\begin{tiny}
\begin{array}{c}
\bar{M}\in\mathbb{R}^{p\times m}
\\
{\rm rank}(\bar{M}) \leq r
\end{array}
\end{tiny}
} \Biggl[ (1+\delta) \|X\bar{M}-X M^*  \|_{F,\Pi}^2 +
\\ 
 \frac{C_1 (1+\delta)^2}{\delta} \frac{ \left( 4 r (m+p+2) \log \left( 1+\frac{\| X \|_F \| \bar{ M } \|_F}{ \sqrt{C_1}} \sqrt{ \frac{nmp}{r(m+p)}} \right) +(m+p) + 2 \log \frac{2}{\varepsilon_n} \right)}{ n} \Biggr] \Biggr\}.
\end{multline*}
Then
$$ \mathbb{E} \Bigl[ \mathbb{P}_{M\sim \hat{\rho}_{\lambda}} (M\in\mathcal{E}_n) \Bigr] \geq 1-\varepsilon_n \xrightarrow[n\rightarrow \infty]{} 1.  $$
\end{theorem}

\section{Computational approximation implementation}
\label{sc_LMC}

\subsection{Langevin Monte Carlo algorithm}
In this section, we propose to compute an approximation of the (quasi) posterior by a suitable version of the LMC algorithm, a gradient-based sampling method. First, we write the logarithm of the density of the posterior
$$
\log  \hat{\rho}_{\lambda}(M)
= 
- \frac{\lambda}{n} \sum_{i=1}^n ( Y_i - (\Pi_C(XM))_{\mathcal{I}_i} )^2
- 
\frac{p+m+2}{2} \log \det (\tau^2 \mathbf{I}_{m} + MM^\intercal ) .
$$

Let us now differentiate this expression in $M$. Note that the term $( Y_i - (\Pi_C(XM))_{\mathcal{I}_i} )^2$ does actually not depend on $M$ locally if $(XM)_{\mathcal{I}_i}\notin [-C,C]$, in this case its differential with respect to $M$ is $\mathbf{0}_{p \times m}$. Otherwise, $( Y_i - (\Pi_C(XM))_{\mathcal{I}_i} )^2=( Y_i - (XM)_{\mathcal{I}_i} )^2$
In order to be able to differentiate the term $(XM)_{\mathcal{I}_i}$, let us introduce a notation for the entries of $\mathcal{I}_i$: $\mathcal{I}_i =(a_i,b_i)$.
Then
$ \nabla (XM)_{\mathcal{I}_i} =D $ where the matrix $D\in\mathbb{R}^{p\times m}$ satisfies $D_{x,y} = \mathbf{1}_{\{x=b_j\}} X_{a_j,y}$. Then
$$
\nabla  \log   \hat{\rho}_{\lambda}( M)
= 
\frac{2\lambda}{n} \sum_{i=1}^n (\nabla (XM)_{\mathcal{I}_i}) ( Y_i - (XM)_{\mathcal{I}_i} ) \mathbf{1}_{\{ |(XM)_{\mathcal{I}_i}| < C \}}
- 
(p+m+2) (\tau^2 \mathbf{I}_{m} + MM^\intercal )^{-1} M .
$$

In this work, we use a constant step-size unadjusted LMC algorithm, see \cite{durmus2019high} for more details. The algorithm is given by an initial matrix $M_0$ and the recursion
\begin{equation}
\label{langevinMC}
M_{k+1}  =  M_{k} - h\nabla \log    \hat{\rho}_{\lambda}(M_k) +\sqrt{2h}\,W_k\qquad k=0,1,\ldots
\end{equation}
where $h>0$ is the step-size and $ W_0, W_1,\ldots$ are independent random matrices with i.i.d. standard Gaussian entries. We provide a pseudo-code for LMC in Algorithm \ref{lmc_algorithm}.

\begin{algorithm}{}
	\caption{LMC for BRRR}
	\begin{algorithmic}[0]
		\State \textbf{Input}: The data $(Y_i,\mathcal{I}_i)$ for $1\leq i\leq n$ and the matrix $X$
		\State \textbf{Parameters}: Positive real numbers $ \lambda, \tau,h,T $ 
		\State \textbf{Output}: The matrix $\hat{M} $
		\State \textbf{Initialize}: $ M_0 ,  \hat{M} = \mathbf{0}_{m\times p}$ 
		\For{$k \gets 1$ to $T$} 
		\State Sample $ M_{k} $ from \eqref{langevinMC};
		\State  $ \hat{M} \gets  \hat{M} +  M_{k}/T $
		\EndFor
	\end{algorithmic}
	\label{lmc_algorithm}
\end{algorithm}

In the case of very large $ p $, the LMC algorithm in \eqref{langevinMC} still requires to calculate a $p\times p$ matrix inversion at each iteration, this might be expensive and can slow down the algorithm. Therefore, we could replace this matrix inversion by its accurately approximation through a convex optimization.  It is noted that the matrix
$\mathbf{B} := (\tau^2 \mathbf{I}_{m} + MM^\intercal )^{-1} M $ is the solution to the following convex optimization problem
$$
\min_{\mathbf{B} } \big\{\| \mathbf{I}_p- M^\top \mathbf{B}  \|_F^2 + \tau^2\|\mathbf{B} \|_F^2\big\}.
$$
The solution of this optimization problem can be conveniently obtained by using the package `\texttt{glmnet}' \cite{glmnet} (with the family option `\texttt{mgaussian}'). This avoids to perform matrix inversion or other costly calculation.  However,  we note here that the LMC algorithm is being used with approximate gradient evaluation, theoretical assessment of this approach can be found in \cite{dalalyan2019user}.

\begin{remark}
For small values of the step-size $ h $, the output of Algorithm \ref{lmc_algorithm}, $ \hat{M} $,   is very close to the integral \eqref{equat_estimator} of interest. However, for some $h$ that may not small enough, the Markov process can be transient and as a consequence the sum explodes \cite{roberts2002langevin}.  Several strategies are available to address this issue: one can take a smaller h and restart the algorithm or a Metropolis–Hastings correction can be included in the algorithm. The Metropolis–Hastings approach ensures the convergence to the desired distribution, however,  the algorithm is greatly slowed down because of an additional acception/rejection step at each iteration. Taking a smaller $h $ also slows down the algorithm but we keep some control on its time of execution.
\end{remark}

\subsection{A Metropolis-adjusted Langevin algorithm}

\begin{algorithm}{}
	\caption{MALA for BRRR}
	\begin{algorithmic}[0]
		\State \textbf{Input}: The data $(Y_i,\mathcal{I}_i)$ for $1\leq i\leq n$ and the matrix $X$
		\State \textbf{Parameters}: Positive real numbers $ \lambda, \tau,h,T $
		\State \textbf{Output}: The matrix $\hat{M} $
		\State \textbf{Initialize}: $M_0 ;  \hat{M} = \mathbf{0}_{m\times p}$ 
		\For{$k = 1$ to $T$} 
		\State Sample $\tilde{M}_{k} $ from \eqref{mala}.
		\State Set $M_k = \tilde{M}_{k}$ with probability $ A_{MALA} $, from \eqref{eq_ac_mala}, otherwise $M_k = M_{k-1}$ .
		\State  $\hat{M} \gets  \hat{M} +  M_{k}/T $ .
		\EndFor
	\end{algorithmic}
	\label{mala_algoritm}
\end{algorithm}

Here, we propose a Metropolis-Hasting correction to the Algorithm \ref{lmc_algorithm}. This approach guarantees the convergence to the (quasi) posterior and it also provides a useful way for choosing $ h $. More precisely, we consider the update rule in \eqref{langevinMC} as a proposal for a new state,
\begin{align}
\tilde{M}_{k+1} = M_{k} - h\nabla \log    \hat{\rho}_{\lambda}  (M_k) +\sqrt{2h}\,W_k,\qquad
k=0,1,\ldots,
\label{mala}
\end{align}
Note that the matrix $\tilde{M}_{k+1} $ is normally distributed with mean $ M_{k} - h\nabla \log    \hat{\rho}_{\lambda}  (M_k) $ and the covariance matrices equal to $ 2h $ times the identity matrices. This proposal is accepted or rejected according to the Metropolis-Hastings algorithm that the proposal is accepted with probability:
\begin{equation}
\label{eq_ac_mala}
A_{MALA} 
:=
\min \left\lbrace 1,  \frac{  \hat{\rho}_{\lambda}  (\tilde{M}_{k+1}) q(M_k | \tilde{M}_{k+1}) }{
	\hat{\rho}_{\lambda}  (M_k ) q(\tilde{M}_{k+1} | M_k ) } \right\rbrace,
\end{equation}
where 
$$
q(x' | x) \propto \exp \left(-\frac{1}{4h}\|x'-x + h\nabla \log    \hat{\rho}_{\lambda}  (x) \|^2_F \right)
$$
is the transition probability density from $x$ to $x'$. The details of the Metropolis-adjusted Langevin algorithm (denoted by MALA) are presented Algorithm \ref{mala_algoritm}. Compared to random-walk Metropolis–Hastings,  the advantage of MALA is that it usually proposes moves into regions of higher probability, which are then more likely to be accepted.

We note that the step-size $h$ is chosen such that the acceptance rate is approximate $0.5$ following \cite{roberts1998optimal}. See Section \ref{sc_simu} for some choices in special cases in our simulations.

\section{Numerical studies}
\label{sc_simu}
\subsection{Simulation setups}

First, we perform some numerical studies on simulated data to access the performance of our proposed algorithms. The experiments are carried by using the \texttt{R} statistical software \cite{r_software}. We compare our algorithms LMC, MALA to the `mRRR' method in \cite{luo2018leveraging}. The `mRRR' method is a frequentist approach where the rank is selected by using 5-fold cross validation. The mRRR method is available from the \texttt{R} package `\texttt{rrpack}' \cite{rrpack}. 

We consider the following model setting. Setting I is a low-dimensional setup with 
$
 \ell = 100,  p = 8, m = 12 
$ 
and the true rank $r = {\rm rank}(M^*) = 2 $. The design matrix $X$ is generated from $ \mathcal{N}(0, \Sigma) $ where the covariance matrix $\Sigma$ is with diagonal entries 1 and off-diagonal entries $\rho_X \geq 0 $, we consider $ \rho_X = 0 $ and $ \rho_X = 0.5 $. Following \cite{luo2018leveraging}, the true coefficient matrix is generated as 
$
 M^* = M_1M_2^\top
$
where 
$ 
M_1 \in \mathbb{R}^{p\times r}
$ 
is an orthogonal matrix from the QR decomposition of a random $ p\times r $ matrix filled with $ \mathcal{N}(0, 1) $ entries, and all entries in 
$ 
 M_2 \in \mathbb{R}^{m\times r} 
$ are i.i.d sampled from $ \mathcal{N}(0, 1) $. The full response matrix $ Z $, as in \cite{luo2018leveraging}, is generated as
$ 
Z = \mathbf{I}_{\ell\times p } + XM^* + E 
$
where $E $ is matrix with entries sampled from $ \mathcal{N}(0, 1) $.

Setting II is similar to Setting I, however, we consider a higher dimensional setup with $ \ell = 500, p = 40, m = 40 $. Setting III is an approximate low-rank set up: This series of simulation is similar to the Setting I, except that the true coefficient is no longer rank 2, but it can be well approximated by a rank 2 matrix:
$
M^* = 2\cdot M_1 M_2^\top + E,
$
where $E $ is matrix with entries sampled from $ \mathcal{N}(0, 0.1) $. Setting IV is also similar to Setting I, however we consider a heavy tail noise for the model \eqref{main_model} where the noise is now assumed to follow a $ t $-Student distribution with 3 degrees of freedom.

In each simulation run, once the full data matrices $ X,Z $ are simulated, we randomly sample $ \vartheta\%  $ entries in $ Z $ with, or without replacement: we let $\mathcal{I}_1,\dots,\mathcal{I}_n$ denote the $n$ observed entries, $Y_i = \mathcal{I}_i$, $\Omega = \{\mathcal{I}_1,\dots,\mathcal{I}_n\}$ the set of observed entries, $\bar\Omega = \{1,\dots,\ell \}\times  \{1,\dots,p\}\setminus \Omega$ the set of non-observed entries and $|\bar\Omega|$ its cardinality.
Comparison of the results of our method with and without replacement are provided in Appendix~\ref{ap_simu}, they are actually quite comparable. On the other hand, the `mRRR' is only implemented in the case without replacement, so we focus the rest of this section to the case without replacement.

These entries are set as missing values that $ \vartheta \in \{ 20, 50, 80 \} $. Under each setting, the entire data generation process is replicated 100 times. The evaluations are done by using the estimation error (Est) and the prediction error (Pred) as
$$
{\rm Est} := \frac{1}{p\ell}\| XM^* - X\hat{M} \|^2_F,
\quad
{\rm Pred} := \sum_{(i,j) \in \bar\Omega} (Z_{i,j} - (X\hat{M} )_{i,j} )^2/ |\bar\Omega|.
$$

The LMC, MALA are run with $5000$ iterations and we take the first $ 2000 $ steps as burn-in. We fixed $\tau^2 = 10 $ and $ \lambda = n/2 $ in all models. The choice of the step-size parameters is in the form $ h = (pm)^{-\alpha} $ where $ \alpha $ is chosen such that the acceptance rate of MALA algorithm is between 0.4 and 0.6. This interval is chosen to enclose 0.574, the optimum acceptance probability for the Langevin MC algorithm \cite{roberts1998optimal}. For example: $h = (pm)^{-1.15} $ for Setting I ($ \rho_X =0, \vartheta = 50\% $),  $h = (pm)^{-1} $ for Setting III ($ \rho_X =0.5, \vartheta = 80\% $).

\subsection{Simulation results}

The results from the simulations are given in Tables \ref{tb_model1}, \ref{tb_model2}, \ref{tb_model3} and \ref{tb_model4}. In general, these 3 methods yield comparable results in Setting I and II, where as the mRRR method is a bit better with less missing data. However, the case of approximate low-rank matrix in Setting III shows that our proposed methods, especially MALA, return better results compared to mRRR. In Setting IV, where our theoretical results do not cover, we observe that mRRR is more stable than our approach.

Looking at the performance of LMC versus MALA, we see that in all settings LMC yields comparable results with MALA while its running time is faster than MALA (as do not require MH correction). However, when the dimension increasing, LMC can return a better result compared to MALA, as in Table \ref{tb_model3}. This can be explained as that LMC converge faster to an approximate distribution whose the mean could be very closed to the target distribution's mean. See also Table \ref{tb_galaxy}. MALA is usually the best method with smallest errors. Moreover, as shown as an example in Table \ref{tb_change_lambda_credible}, MALA comes with reliable credible intervals while the credible intervals output from LMC have smaller coverage rate.

We further access the affect of changing $ \lambda $ in our proposed algorithms LMC and MALA, especially we emphasize on the empirical coverage of the 95\% credible interval. The LMC, MALA are run with 10000 iterations and we take the first $ 2000 $ steps as burn-in. The results are given in Table \ref{tb_change_lambda_credible} where we examine Setting III for the case of $ \rho_X = 0 $ and the missing rate in $ Y $ is $ \vartheta = 50\% $. The message from these results is that taking smaller $ \lambda $ could improve the empirical coverage rate, however the cost is in increasing the mean squared error. We also see that the value $ \lambda = n/2 $ is good choice that we have used in all our simulations. Some results for different values of $ \tau $ in the prior are given in Table \ref{tb_change_tau_prior}. We see that the different values of $ \tau $ will effect differently on LMC and MALA thus we stick a fixed value of $ \tau $ in our simulations. It is worth noted that the choices of $ \tau $ and $ \lambda $ are greatly effect to the values of the step-size $ h $.

\section{Real data application}

\subsection{Galaxy COMBO17 data: imcomplete data again full data}
This particular subset of the COMBO17 data set, see \cite{izenman2008modern}, consists of the n = 3,438 objects in the Chandra Deep Field South that are classified as “Galaxies” and for which there are no missing values. We also omitted five redundant variables and all error variables in the data set; the 29 remaining variables were then divided into a group of 23 variables as predictor matrix $ X $ and a group of 6 variables as response matrix $ Y $, details can be found in Table 7.2 in \cite{izenman2008modern}.

Here we emphasize that the response matrix $ Z $ is fully observed. The study in this section aims at finding how the missing rate in the response matrix affects to the estimation procedure. As this dataset contains data at different scale, we further rescale the data so that its columns has 0-mean and 1-variance. We first estimate, by using the OLS (ordinary least squares) method, the parameter matrix given fully observed data $ Z , X $ and denote this estimation by $ \hat{M}_{OLS} $. Then, we randomly remove  $ \vartheta = 20\%, 50\% $ and $ 80\% $ entries of $ Z $ and re-estimate the parameter matrix. This procedure is repeated 100 times and we report the average results. The LMC, MALA are run with 15000 iterations and we take the first $ 5000 $ steps as burn-in. 

To access the accuracy, we consider the estimation error and the prediction error as
\begin{align*}
{\rm Est} &:=  \| X\hat{M}_{OLS} -  X \hat{M} \|^2_F/ (\ell p ),
\quad
{\rm Pred} := \sum_{(i,j) \in \bar\Omega} (Z_{i,j} - (X\hat{M} )_{i,j} )^2/ |\bar\Omega|.
\end{align*}

The result outputs are given in Table \ref{tb_galaxy}. It is interesting to see that LMC here is the method returns with smallest errors for the cases of 20\% and 50\% missing rate. However, when the missing rate is 80\% MALA method is better than LMC. In all cases, mRRR method comes with higher errors compared with LMC and MALA.

\subsection{Yeast cell cycle data with imcomplete response}
The dataset is available in the `\texttt{secure}' \texttt{R}-package \cite{mishra2017sequential}. The analysis of the yeast cell cycle enables us to identify transcription factors (TFs) which regulate ribonucleic acid (RNA) levels within the eukaryotic cell cycle. The dataset contains two components: the chromatin immunoprecipitation (ChIP) data and eukaryotic cell cycle data. The binding information of a subset of 1790 genes and 113 TFs was included in the ChIP data \cite{lee2002transcriptional}, which results in the design matrix $ X $ of dimension $ 1790 \times 113 $. The cell cycle data were obtained by measuring the RNA levels every 7 minutes for 119 minutes, thus a total of 18 time points. The resulting response matrix $ Z $ is of dimension  $ 1790 \times 18 $. There are 626 missing entries in the response matrix $ Z $ which is around $ 1.94\% $. We further rescale the data so that its columns has 0-mean and 1-variance.

The data are randomly divided into a training set and a test set of size $ 80\% $ and $ 20\% $ with respect to the sample size. Model estimation is done by using the training data. Then the predictive performance is calculated on the test data by its mean squared prediction error $ \| Z_{test} - X_{test}\hat{B} \|_F^2 / ( p \ell_{test}) $ (the missing entries are discarded in calculation error), where $ ( Z_{test} , X_{test}) $ denotes the test set. We repeat the random training/test splitting process $100$ times and report the average mean squared prediction error for each method. We ran LMC and MALA with 30000 iterations and the first 5000 steps are ignored as burn-in period. The results are given in Figure \ref{tb_cellcycle}, we can see that LMC and MALA yield similar results, while mRRR method returns slightly smaller prediction error than LMC and MALA.

\section{Conclusion}
We proposed a quasi-Bayesian estimation method for reduced rank regression with incomplete responses. We proved the contraction of the posterior and the convergence of the posterior mean at the optimal rate. We implemented our algorithm via LMC and MALA, with promising results. Some questions remain open, such as missing entries in the covariate matrix $X$, and theoretical properties for the case of sampling without replacement might be the object of future works.

\section*{Acknowledgments}
TTM is supported by the Norwegian Research Council, grant number 309960 through the Centre for Geophysical Forecasting at NTNU.

\section*{Availability of data and materials}%% if any
The R codes and data used in the numerical experiments are available at:  \url{https://github.com/tienmt/BRRR_missing} .

%tb_cellcycle

\begin{figure}[!h]
	\centering
	\includegraphics[scale=.5]{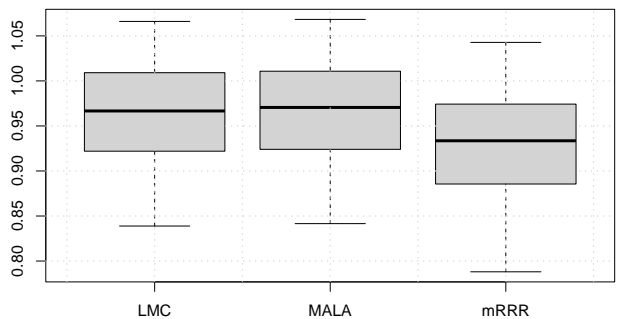}
	\caption{ Boxplot on prediction error for Cell cylcle real data for different methods}
	\label{tb_cellcycle}
\end{figure}

%tb_model1
\begin{table}[!h]
	\caption{Simulation results on simulated data in Setting I ($ \ell = 100,  p = 8, m = 12 $ and $r  = 3$) for different methods, with their standard error in parentheses.  (Est: average of estimation error; Pred: average of prediction error).}
	\centering
	\small
	\makebox[\textwidth]{
		\begin{tabular}{ | p{11mm}|ccc|| ccc | } 
			\toprule
			& \multicolumn{3}{ c|| }{$\rho_X = 0,  \vartheta = 20\%$} 
			& \multicolumn{3}{ c | }{$\rho_X = 0.5 , \vartheta = 20\%$} 
			\\
	Errors    	& LMC & MALA & mRRR & LMC & MALA & mRRR
			\\ 
			\midrule
Est 	& 0.165 (0.026) & 0.163 (0.026) & \textbf{0.072} (0.038)
		& 0.157 (0.023) & 0.156 (0.024) & \textbf{0.083} (0.047)
			\\ 
Pred 	& 2.234 (0.236) & 2.233 (0.234) & \textbf{2.099} (0.216)
		& 2.174 (0.206) & 2.171 (0.210) & \textbf{2.091} (0.208)
			\\  
			\midrule
			& \multicolumn{3}{ c || }{$\rho_X = 0, \vartheta = 50\%$} 
			& \multicolumn{3}{ c | }{$\rho_X = 0.5  , \vartheta = 50\%$} 
			\\ 
			\midrule
Est 	& 0.293 (0.052) & \textbf{0.288} (0.051) & 0.315 (0.278)
		& 0.295 (0.050) & \textbf{0.286} (0.047) & 0.366 (0.210)
			\\
Pred  	& 2.322 (0.160) & \textbf{2.313} (0.157) & 2.315 (0.338)
		& 2.349 (0.169) & \textbf{2.338} (0.161) & 2.402 (0.281)
			\\  
			\midrule
			& \multicolumn{3}{ c || }{$\rho_X = 0, \vartheta = 80\%$} 
			& \multicolumn{3}{ c | }{$\rho_X = 0.5  , \vartheta = 80\%$} 
			\\ 
			\midrule
Est 	& 14.41 (26.59) & \textbf{1.689} (0.507) & 1.707 (0.614)
		& 5.057 (4.356) & 1.568 (0.462) & \textbf{1.328} (0.625)
			\\
Pred  	& 19.79 (32.99) & 3.966 (0.632) & \textbf{3.718} (0.653)
		& 8.147 (5.418) & 3.832 (0.585) & \textbf{3.345} (0.661)
			\\ \bottomrule
		\end{tabular}
	}
	\label{tb_model1}
\end{table}

%tb_model2
\begin{table}[!h]
	\caption{Simulation results on simulated data in Setting II ($ \ell = 500,  p = 40, m = 40 $ and $r  = 3$) for different methods, with their standard error in parentheses.  (Est: average of estimation error; Pred: average of prediction error).}
	\centering
	\makebox[\textwidth]{
		\small
		\begin{tabular}{ | p{11mm}|ccc|| ccc | } 
			\toprule
			& \multicolumn{3}{ c|| }{$\rho_X = 0,  \vartheta = 20\%$} 
			& \multicolumn{3}{ c | }{$\rho_X = 0.5 , \vartheta = 20\%$} 
			\\
Errors   	& LMC & MALA & mRRR & LMC & MALA & mRRR
			\\ \midrule
Est		& 0.104 (0.004) & 0.104 (0.004) & \textbf{0.011} (0.001)
		& 0.107 (0.004) & 0.114 (0.004) & \textbf{0.066} (0.077)
			\\ 
Pred 	& 2.118 (0.038) & 2.119 (0.038) & \textbf{2.015} (0.036)
		& 2.116 (0.041) & 2.123 (0.042) & \textbf{2.062} (0.088)
			\\  
			\midrule
			& \multicolumn{3}{ c || }{$\rho_X = 0, \vartheta = 50\%$} 
			& \multicolumn{3}{ c | }{$\rho_X = 0.5  , \vartheta = 50\%$} 
			\\ \midrule
Est		& 0.179 (0.006) & 0.178 (0.006) & \textbf{0.017} (0.002)
		& 0.186 (0.007) & 0.198 (0.008) & \textbf{0.074} (0.068)
			\\
Pred 	& 2.195 (0.028) & 2.195 (0.027) & \textbf{2.019} (0.025)
		& 2.203 (0.028) & 2.217 (0.028) & \textbf{2.075} (0.071)
			\\  
			\midrule
			& \multicolumn{3}{ c || }{$\rho_X = 0, \vartheta = 80\%$} 
			& \multicolumn{3}{ c | }{$\rho_X = 0.5  , \vartheta = 80\%$} 
			\\ \midrule
Est 	& 0.663 (0.032) & \textbf{0.632} (0.031) & 1.902 (0.367)
		& 0.742 (0.037) & \textbf{0.723} (0.036) & 1.483 (0.660)
			\\
Pred  	& 2.728 (0.045) & \textbf{2.691} (0.042) & 3.899 (0.371)
		& 2.814 (0.054) & \textbf{2.790} (0.051) & 3.480 (0.658)
			\\ \bottomrule
		\end{tabular}
	}
	\label{tb_model2}
\end{table}

%tb_model3
\begin{table}[!h]
	\caption{Simulation results on simulated data in Setting III, approximate low-rank model, for different methods, with their standard error in parentheses.  (Est: average of estimation error; Pred: average of prediction error).}
	\centering
	\small
	\makebox[\textwidth]{
		\begin{tabular}{ | p{11mm} | ccc || ccc | } 
			\toprule
			& \multicolumn{3}{ c|| }{$\rho_X = 0,  \vartheta = 20\%$} 
			& \multicolumn{3}{ c | }{$\rho_X = 0.5 , \vartheta = 20\%$} 
			\\
Errors    	& LMC & MALA & mRRR 
			& LMC & MALA & mRRR
			\\ \midrule
Est 	& 0.157 (0.026) & \textbf{0.156} (0.025) & 0.313 (0.671)
		& 0.161 (0.026) & \textbf{0.160} (0.025) & 0.322 (0.545)
			\\ 
Pred  	& 2.184 (0.220) & \textbf{2.182} (0.218) & 2.277 (0.552)
		& \textbf{2.193} (0.222) & 2.198 (0.226) & 2.361 (0.638)
			\\ 
			\midrule
			& \multicolumn{3}{ c || }{$\rho_X = 0, \vartheta = 50\%$} 
			& \multicolumn{3}{ c | }{$\rho_X = 0.5  , \vartheta = 50\%$} 
			\\ \midrule
Est 	& 0.286 (0.050) & \textbf{0.281} (0.049) & 0.563 (0.660)
		& 0.293 (0.051) & \textbf{0.288} (0.050) & 0.630 (0.698)
			\\
Pred  	& 2.304 (0.182) & \textbf{2.298} (0.103) & 2.565 (0.733)
		& 2.337 (0.181) & \textbf{2.332} (0.179) & 2.631 (0.731)
			\\ 
			\midrule
			& \multicolumn{3}{ c || }{$\rho_X = 0, \vartheta = 80\%$} 
			& \multicolumn{3}{ c | }{$\rho_X = 0.5 , \vartheta = 80\%$} 
			\\ \midrule
Est 	& 7.766 (15.18) & \textbf{1.708} (0.683) & 2.985 (1.971)
		& 5.224 (4.204) & \textbf{1.645} (0.523) & 2.444 (1.500)
			\\
Pred 	& 11.51 (18.81) & \textbf{3.985} (0.907) & 5.097 (2.069)
		& 8.351 (5.232) & \textbf{3.918} (0.660) & 4.505 (1.509)
			\\ 
			\bottomrule
		\end{tabular}
	}
	\label{tb_model3}
\end{table}

%tb_model4
\begin{table}[!h]
	\caption{Simulation results on simulated data in Model IV, heavy tail noise ($ \ell = 100,  p = 8, m = 12 $ and $r  = 3$) for different methods, with their standard error in parentheses.  (Est: average of estimation error; Pred: average of prediction error).}
	\centering
	\small
	\makebox[\textwidth]{
		\begin{tabular}{ | p{11mm}|ccc| ccc | } 
			\toprule
			& \multicolumn{3}{ c| }{$\rho_X = 0,  \vartheta = 20\%$} 
			& \multicolumn{3}{ c | }{$\rho_X = 0.5 , \vartheta = 20\%$} 
			\\
Errors   & LMC & MALA & mRRR & LMC & MALA & mRRR
			\\ \midrule
Est 	& 0.470 (0.149) & 0.465 (0.148) & \textbf{0.416} (0.353)
		& 0.460 (0.159) & 0.455 (0.158) & \textbf{0.311} (0.190)
			\\ 
Pred  	& 4.439 (1.305) & 4.435 (1.305) & \textbf{4.355} (1.348)
		& 4.379 (1.273) & 4.374 (1.271) & \textbf{4.159} (1.274)
			\\  
			\midrule
			& \multicolumn{3}{ c | }{$\rho_X = 0, \vartheta = 50\%$} 
			& \multicolumn{3}{ c | }{$\rho_X = 0.5  , \vartheta = 50\%$} 
			\\ 
			\midrule
Est 	& 0.829 (0.230) & 0.812 (0.225) & \textbf{0.681} (0.341)
		& 0.862 (0.372) & 0.840 (0.365) & \textbf{0.517} (0.213)
			\\
Pred  	& 4.799 (1.340) & 4.777 (1.334) & \textbf{4.533} (1.306)
		& 4.900 (1.688) & 4.870 (1.685) & \textbf{4.455} (1.521)
			\\  
			\midrule
			& \multicolumn{3}{ c | }{$\rho_X = 0, \vartheta = 80\%$} 
			& \multicolumn{3}{ c | }{$\rho_X = 0.5  , \vartheta = 80\%$} 
			\\ \midrule
Est 	& 11.65 (12.26) & 4.706 (3.256) & \textbf{1.457} (0.593)
		& 7.905 (4.779) & 3.986 (1.870) & \textbf{1.012} (0.458)
			\\
Pred  	& 18.09 (15.31) & 9.420 (4.068) & \textbf{5.546} (1.772)
		& 13.32 (5.885) & 8.510 (2.374) & \textbf{4.986} (0.961)
			\\ \bottomrule
		\end{tabular}
	}
	\label{tb_model4}
\end{table}

%tb_change_lambda_credible
\begin{table}[!h]
	\caption{Results on credible intervals for changing $ \lambda $ in Setting III  $ \rho_X = 0 , \vartheta = 50\% $, with their standard error over 100 repetitions in parentheses. (ECovR is the empirical coverage rate, MSE is the mean squared error). }
	\centering
	\small
	\begin{tabular}{ |  p{1.5cm} p{1.5cm} |c  c  c  c|  } 
		\toprule 
 & & $ \lambda = n $ & $ \lambda = n/2 $ & $ \lambda = n/8 $ & $ \lambda = n/32 $
		\\ \midrule
\multirow{2}{*}{LMC} &
MSE &  0.279 (0.046) & 0.280 (0.046) & 0.288 (0.049) & 0.347 (0.064)
\\
	&ECovR & 0.776 (0.047) & 0.911 (0.032) & 0.999 (0.003) & 1.0 (0.0)
		\\ \midrule 
		&&&&&
		\\[-.7em]  \midrule
\multirow{2}{*}{MALA} 	&
MSE & 0.277 (0.045) & 0.276 (0.045) & 0.276 (0.046)& 0.323 (0.065)
\\
	&ECovR & 0.840 (0.041) & 0.950 (0.026) & 1.000 (0.001)& 1.0 (0.0)
		\\ \bottomrule
	\end{tabular}
	\label{tb_change_lambda_credible}
\end{table}

%tb_change_tau_prior
\begin{table}[!h]
	\caption{Results for changing $ \tau $ in Setting I with $ \rho_X = 0 , \vartheta = 50\% $, with their standard error over 100 repetitions in parentheses. (ECovR is the empirical coverage rate, MSE is the mean squared error.)}
	\centering
	\small
	\begin{tabular}{ | p{1.5cm} p{1.5cm} |c  c  c  c|  } 
		\toprule
	&	& $ \tau = 10 $ & $ \tau = 5 $ & $ \tau = 1 $ & $ \tau = 0.1 $
		\\ \midrule
\multirow{2}{*}{LMC} &
		MSE 
		&  0.287 (0.044) & 0.291 (0.045) & 0.393 (0.060) & 1.415 (1.899)
		\\
&		ECovR 
		& 0.909 (0.034) & 0.905 (0.033) & 0.858 (0.038) & 0.575 (0.072)
		\\ \midrule
		&&&&&
		\\[-.8em] \midrule
\multirow{2}{*}{MALA} 	&	
	MSE & 
		0.283 (0.044) & 0.278 (0.044) & 0.224 (0.037)& 2.508 (4.473)
		\\ 
	&	ECovR &  
		0.949 (0.023) & 0.949 (0.022) & 0.955 (0.024)& 0.459 (0.095)
		\\ \bottomrule
	\end{tabular}
	\label{tb_change_tau_prior}
\end{table}

%tb_galaxy
\begin{table}[!h]
	\caption{Results with real Galaxy data for different methods, with their standard error in parentheses.  (Est: average of estimation error; Pred: average of prediction error.}
	\centering
	\small
\begin{tabular}{ | p{15mm} | p{15mm}|c|c | } 
\toprule	
missing rate & method & Est & Pred 
\\	\midrule			
\multirow{3}{*}{$ \vartheta = 20\%$} 	
& LMC 	& \textbf{0.023} (0.004)  & \textbf{0.633} (0.050)
	\\
& MALA	& 0.030 (0.005)  & 0.639 (0.054) 
	\\
& mRRR  & 0.117 (0.028)  & 0.704 (0.056) 
	\\ \hline 
	& & &
	\\[-.7em]	\midrule
\multirow{3}{*}{$ \vartheta = 50\%$}		
& LMC 	& \textbf{0.051} (0.021)  & \textbf{0.670} (0.049) 
	\\
& MALA	& 0.057 (0.021)  & 0.673 (0.049) 
	\\
& mRRR 	& 0.131 (0.030)  & 0.724 (0.060) 
	\\ \hline 
	& & &
	\\[-.7em]	\midrule
\multirow{3}{*}{$ \vartheta = 80\%$}		
& LMC 	& 0.213 (0.106)  & 0.844 (0.134)
	\\
& MALA	& \textbf{0.185} (0.097)  & \textbf{0.803} (0.123)
	\\
& mRRR 	& 0.341 (0.006)	 & 0.914 (0.012)
	\\ 
	\bottomrule
\end{tabular}
\label{tb_galaxy}
\end{table}

\clearpage
\appendix
\section{Appendix: proofs}
\label{sc_appendix_proof}

As mentioned above, our strategy for the proofs is to use oracle type PAC-Bayes bounds, in the spirit of ~\cite{catoni2004statistical}. We start with a few preliminary lemmas, we then provide the proof of Theorem \ref{thrm_main} and the proof of Theorem \ref{thrm_contraction}.

\subsection{Preliminary results}

First, we state a version of Bernstein's inequality taken from \cite{MR2319879}, (2.21) in Proposition 2.9 page 24.
\begin{lemma}[Bernstein's inequality]
\label{lemmemassart} Let $U_{1}$, \ldots, $U_{n}$ be independent real
valued random variables. Let us assume that there are two constants
$v$ and $w$ such that
$
 \sum_{i=1}^{n} \mathbb{E}[U_{i}^{2}] \leq v 
$
and for all integers $k\geq 3$,
$
 \sum_{i=1}^{n} \mathbb{E}\left[(U_{i})^{k}\right] \leq v\frac{k!w^{k-2}}{2}. 
$
Then, for any $\zeta\in (0,1/w)$,
$$ \mathbb{E}
\exp\left[\zeta\sum_{i=1}^{n}\left[U_{i}-\mathbb{E}(U_{i})\right]
\right]
        \leq \exp\left(\frac{v\zeta^{2}}{2(1-w\zeta)} \right) .$$
\end{lemma}
Another basic tool to derive PAC-Bayes bounds is Donsker and Varadhan's variational inequality, that we recap here. We refer to Lemma 1.1.3 in Catoni \cite{catonibook} for a proof (among others). From now, for any $\Theta\subset \mathbb{R}^{n_1}$ or $\Theta\subset \mathbb{R}^{n_1 \times n_2}$, we let $\mathcal{P}(\Theta)$ denote the set of all probability distributions on $\Theta$ equipped with the Borel $\sigma$-algebra. We remind that when $(\mu,\nu)\in \mathcal{P}(\Theta)^2$, the Kullback-Leibler divergence is defined by $\mathcal{K}(\nu,\mu) = \int \log \left( \frac{{\rm d}\nu}{{\rm d}\mu}(\theta) \right) \nu({\rm d}\theta) $ if $\nu$ admits a density $\frac{{\rm d}\nu}{{\rm d}\mu}$ with respect to $\mu$, and $\mathcal{K}(\nu,\mu)=+\infty$ otherwise.
\begin{lemma}[Donsker and Varadhan's variational formula]
\label{lemma:dv}
Let $\mu \in\mathcal{P}(\Theta)$. For any measurable, bounded function $h:\Theta\rightarrow\mathbb{R}$ we have:
\begin{equation*}
\log \int {\rm e}^{h(\theta)} \mu({\rm d}\theta) = \sup_{\rho\in\mathcal{P}(\Theta)}\left[\int h(\theta) \rho({\rm d}\theta) -KL (\rho\|\mu)\right].
\end{equation*}
Moreover, the supremum with respect to $\rho$ in the right-hand side is
reached for the Gibbs measure
$\mu_{h}$ defined by its density with respect to $\pi$
\begin{equation}
\label{equa:def:gibbs:measure}
\frac{{\rm d}\mu_{h}}{{\rm d}\mu}(\theta) =  \frac{{\rm e}^{h(\theta)}}
{ \int {\rm e}^{h(\vartheta)} \mu({\rm d}\vartheta) }.
\end{equation}
\end{lemma}
These two lemmas are the only tools we need to prove Theorem \ref{thrm_main} and Theorem \ref{thrm_contraction}. Their proof is quite similar, with a few differences. In order not to do the same derivation twice, we will state the common parts of the proofs as a separate result: Lemma \ref{lemma:exponential}. Note that the proof of this lemma will already use Lemmas \ref{lemmemassart} and \ref{lemma:dv}.
\begin{lemma}
 \label{lemma:exponential}
 Under Assumptions~\ref{assum_bounded} and~\ref{assum_noise}, put
 \begin{equation}
\label{defalpha}
\alpha = \left(\lambda
-\frac{\lambda^{2} C_1 }{2n(1-\frac{ C_2 \lambda}{n})}\right)
\end{equation}
and
 \begin{equation}
\label{defbeta}
\beta = \left(\lambda
+\frac{\lambda^{2} C_1}{2n(1-\frac{ C_2 \lambda}{n})}\right) .
\end{equation}
Then for any $\varepsilon\in(0,1)$, and $\lambda \in (0,n/C_2)$,
\begin{multline}
  \mathbb{E} \int \exp  \Biggl\{  \alpha    \Bigl( R(M) - R(M^*) \Bigr)
 +\lambda\Bigl( -r(M) + r(M^*) \Bigr)    
    - \log \left[\frac{d\hat{\rho}_{\lambda}}{d \pi} (M)  \right]
         - \log\frac{2}{\varepsilon}
        \Biggr\}
         \hat{\rho}_{\lambda}(d M)
\leq \frac{\varepsilon}{2}
\label{lemma:exponential:1}
\end{multline}
and
\begin{align}
\mathbb{E} \sup_{\rho \in \mathcal{P}(\mathbb{R}^{m\times p}) } \exp\Biggl[  \beta
               \left(-\int Rd\rho + R(M^*) \right)
+ \lambda \left( \int r d\rho - r(M^*) \right) 
-
\mathcal{K}(\rho, \pi) - \log \frac{2}{\varepsilon}\Biggr] \leq
\frac{\varepsilon}{2}.
\label{lemma:exponential:2}
\end{align}
\end{lemma}
\begin{proof}[\textbf{Proof of Lemma~\ref{lemma:exponential}}]
We start by proving the first inequality, that is \eqref{lemma:exponential:1}.
Fix any $M$ with $\|M\|_\infty\leq C $ and put
 $$ T_{i} =   \left( Y_i - (XM^*)_{\mathcal{I}_{i}} \right)^{2}
                    - \left( Y_i - (\Pi_C(XM))_{\mathcal{I}_{i}} \right)^{2} .$$
Note that $T_1,\dots,T_n$ are independent by construction. We have
\begin{align*}
\sum_{i=1}^{n} \mathbb{E}[T_{i}^{2}]  
&  = \sum_{i=1}^{n} \mathbb{E}
\left[
       \left( 2Y_{i} - (XM^*)_{\mathcal{I}_{i}} - (XM)_{\mathcal{I}_{i}} \right)^{2}
\left( (XM^*)_{\mathcal{I}_{i}} - (XM)_{\mathcal{I}_{i}} \right)^2
            \right]
\\
&  = \sum_{i=1}^{n} \mathbb{E} \left[
       \left( 2 \mathcal{E}_{i} + (XM^*)_{\mathcal{I}_i}  - (XM)_{\mathcal{I}_{i}}  \right)^{2}
\left( (XM^*)_{\mathcal{I}_{i}} - (XM)_{\mathcal{I}_{i}} \right)^2
            \right]
\\
&    \leq \sum_{i=1}^{n} \mathbb{E} \left[
        8 \left[ \mathcal{E}_{i}^{2} + C^2 \right]
\left[ (XM^*)_{\mathcal{I}_{i}} -  (XM)_{\mathcal{I}_{i}} \right]^2
            \right]
\\
&     = \sum_{i=1}^{n}  8 \left[ \mathcal{E}_{i}^{2} + C^2 \right]
     \mathbb{E} \left[ (XM^*)_{\mathcal{I}_{i}} - (XM)_{\mathcal{I}_{i}} \right]^2
\\
 &   \leq 8 n( \sigma^2 + C^2 )
\left[ R(M) - R(M^*)\right]
\\
& = n C_1\left[ R(M) - R(M^*)\right]  =:v(M, M^*).
\end{align*}
Next we have, for any integer $k\geq 3$, that
\begin{align*}
\sum_{i=1}^{n} \mathbb{E}\left[(T_{i})^{k}\right] \leq 
&
\sum_{i=1}^{n} \mathbb{E} \left[
       \left| 2Y_{i} - (XM^*)_{\mathcal{I}_{i}} - (XM)_{\mathcal{I}_{i}} \right|^{k}
\left| (XM^*)_{\mathcal{I}_{i}} - (XM)_{\mathcal{I}_{i}} \right|^k
            \right]
\\
\leq & \sum_{i=1}^{n} \mathbb{E} \left[
       2^{k-1} \left[  |2 \mathcal{E}_{i}|^{k} + (2 C)^k  \right]
 \left| (XM^*)_{\mathcal{I}_{i}} - (XM)_{\mathcal{I}_{i}} \right|^{k}
            \right]
\\
\leq  &  \sum_{i=1}^{n} \mathbb{E} \left[
       2^{2k-1}\left(  |\mathcal{E}_{i}|^{k} + C^k \right)
       (2C)^{k-2}
 \left| (XM^*)_{\mathcal{I}_{i}} - (XM)_{\mathcal{I}_{i}} \right|^{2}
            \right]
\\
\leq    &      2^{2k-1} \left[ \sigma^{2}k!\xi^{k-2}
             + C^k   \right]  (2C)^{k-2}
\sum_{i=1}^{n}\mathbb{E}  \left| (XM^*)_{\mathcal{I}_{i}} - (XM)_{\mathcal{I}_{i}} \right|^{2}
\\
\leq  &  \frac{ 2^{3k-3} \left[ \sigma^{2}k!\xi^{k-2} + C^k \right] C^{k-2} }{8(\sigma^2 + C^2) } v(M, M^*)
\\
\leq  &  \frac{ 2^{3k-6}  \left[ \sigma^{2}\xi^{k-2} + C^k \right] C^{k-2}}{(\sigma^2 + C^2) } k! v(M, M^*)
\\
\leq  &  2^{3k-5} \left[ \xi^{k-2} + C^{k-2} \right] C^{k-2} k! v(M, M^*)
\\
\leq  &  2^{3k-4} \max(\xi,C)^{k-2} C^{k-2} k! v(M, M^*)
\\
=  &  [2^3 \max(\xi,C) C ]^{k-2} 2^2 k! v(M, M^*)
\end{align*}
and use the fact that, for any $k\geq 3$, $2^2 \leq 2^{3(k-2)}/2 $ to obtain
\begin{align*}
\sum_{i=1}^{n} \mathbb{E}\left[(T_{i})^{k}\right] 
\leq   
 \frac{[2^6 \max(\xi,C) C ]^{k-2} k! v(M, M^*)}{2}
= 
v(M, M^*) \frac{k!C_2^{k-2}}{2}.
\end{align*}
Thus, we can apply Lemma~\ref{lemmemassart} with $U_i := T_i$, $v := v(M, M^*)$, $w:=C_2$ and $\zeta := \lambda/n$. We obtain, for any $\lambda\in(0,n/w)=(0,n/C_2)$,
\begin{align*}
\mathbb{E} \exp\left[\lambda
\Bigl( R(M)-R(M^*)-r(M)+r(M^*)\Bigr)\right]
 \leq
\exp\left[\frac{v\lambda^{2}}{2n^{2}(1-\frac{w\lambda}{n})}\right]
 =
\exp\left[\frac{ C_1
\left[ R(M) - R(M^*)\right] \lambda^{2}}{2n(1-\frac{C_2 \lambda}{n})}\right].
\end{align*}
Rearranging terms, and using the definition of $\alpha$ (that is~\eqref{defalpha}),
$$
\mathbb{E} \exp\left[\alpha
                      \Bigl( R(M) - R(M^*) \Bigr)
                    +\lambda\Bigl( -r(M) + r(M^*) \Bigr) \right] \leq 1.
$$
Multiplying both sides by $\varepsilon/2$ and then integrating w.r.t. the probability distribution $ \pi(.) $, we get
$$
\int \mathbb{E} \exp\Biggl[\alpha
                      \Bigl( R(M) - R(M^*) \Bigr)
                    +\lambda\Bigl( -r(M) + r(M^*) \Bigr)
         - \log\frac{2}{\varepsilon}\Biggr]  \pi (d M) \leq \frac{\varepsilon}{2}.
$$
Next, Fubini's theorem gives
$$
\mathbb{E} \int  \exp\Biggl[\alpha
                      \Bigl( R(M) - R(M^*) \Bigr)
                    +\lambda\Bigl( -r(M) + r(M^*) \Bigr)
         - \log\frac{2}{\varepsilon}\Biggr]  \pi (d M) \leq \frac{\varepsilon}{2}.
$$
and note that for any measurable function $h$,
\begin{align*}
\int  \exp[h(M)]  \pi (d M)
= 
\int  \frac{\exp[h(M)] } { \frac{{\rm d} \hat{\rho}_{\lambda}}{{\rm d}\pi}(M) } \hat{\rho}_{\lambda}(d M)
= 
\int  \exp\left[h(M)- \log \frac{{\rm d} \hat{\rho}_{\lambda}}{{\rm d}\pi}(M) \right] \hat{\rho}_{\lambda}(d M)
\end{align*}
to get \eqref{lemma:exponential:1}.

Let us now prove \eqref{lemma:exponential:2}. Here again, we start with an application of Lemma~\ref{lemmemassart}, but this time with $U_i := - T_i$ (we keep $v := v(M, M^*)$, $w:=C_2$ and $\zeta := \lambda/n$). We obtain, for any $\lambda\in(0,n/C_2)$,
$$
\mathbb{E}  \exp\left[\lambda
\Bigl( r(M)+r(M^*)-R(M)+R(M^*)\Bigr)\right]
\leq
\exp\left[\frac{ C_1
\left[ R(M) - R(M^*)\right] \lambda^{2}}{2n(1-\frac{C_2 \lambda}{n})}\right].
$$
Rearranging terms, using the definition of $\beta$ (that is~\eqref{defbeta}) and multiplying both sides by $\varepsilon/2$, we obtain
\begin{equation*}
\mathbb{E} \exp\Biggl[\beta
               \left(-R(M) + R(M^*) \right)
+ \lambda \left( r(M)- r(M^*) \right) - \log \frac{2}{\varepsilon}\Biggr] \leq
\frac{\varepsilon}{2}.
\end{equation*}
We integrate with respect to $\pi$ and use Fubini to get:
\begin{equation*}
\mathbb{E} \int \exp\Biggl[\beta
               \left(-R(M) + R(M^*) \right)
+ \lambda \left( r(M)- r(M^*) \right) - \log \frac{2}{\varepsilon}\Biggr] \pi({\rm d} M ) \leq
\frac{\varepsilon}{2}.
\end{equation*}
Here, we use a different argument from the proof of the first inequality: we use Lemma \ref{lemma:dv} on the integral, this gives directly \eqref{lemma:exponential:2}.
\end{proof}
Finally, in both proofs, we will use quite often distributions $\rho\in\mathcal{P}(\mathbb{R}^{m\times p})$ that will be defined as translations of the prior $\pi$. We introduce the following notation.
\begin{dfn}
\label{dfn:posterior:transla}
 For any  matrix $\bar{M} \in \mathbb{R}^{m\times p}$, we define $\rho_{\bar{M}} \in \mathcal{P}(\mathbb{R}^{m\times p})$
 by 
 $$
 \rho_{\bar{M}}(M) = \pi(\bar{M}-M).
 $$
\end{dfn}

\noindent The following technical lemmas, that can be found for example in \cite{dalalyan2020exponential}, will also be useful in the proofs.

\begin{lemma}[Lemma 1 in \cite{dalalyan2020exponential}]
\label{lemma:arnak:1}
 We have
 $ \int \|M\|_{F}^2 \pi({\rm d}M) \leq m p \tau^2. $
\end{lemma}

\begin{lemma}[Lemma 2 in \cite{dalalyan2020exponential}]
\label{lemma:arnak:2}
For any $\bar{M} \in \mathbb{R}^{m\times p}$, we have
$$
 \mathcal{K}( \rho_{\bar{M}} ,  \pi) 
\leq 
2 {\rm rank} (\bar{M}) (m+p+2) \log \left( 1+ \frac{\| \bar{ M } \|_F}{\tau \sqrt{2{\rm rank} (\bar{M})}} \right) $$
with the convention $0\log(1+0/0)=0$.
\end{lemma}

\subsection{Proof of Theorem \ref{thrm_main}}

\begin{proof}[\textbf{Proof of Theorem \ref{thrm_main}}]
We apply Lemma~\ref{lemma:exponential}, and will start to work on its first inequality~\eqref{lemma:exponential:1}. An application of Jensen's inequality yields
$$
\mathbb{E} \exp\Biggl[\alpha
                      \left( \int R d \hat{\rho}_{\lambda} - R(M^*) \right)
                    +\lambda\left( -\int r d \hat{\rho}_{\lambda} + r(M^*) \right)
                 - \mathcal{K}(\hat{\rho}_{\lambda}, \pi)
         - \log\frac{2}{\varepsilon}\Biggr] \leq \frac{\varepsilon}{2}.
$$
We now use the standard Chernoff's trick to transform an exponential moment inequality into a deviation inequality, that is: $\exp(x) \geq \mathbf{1}_{\mathbb{R}_{+}}(x)$, we obtain
\begin{align}
\mathbb{P}\Biggl\{ \Biggr[ \alpha
                      \left(\int R d\hat{\rho}_{\lambda} - R(M^*)\right)
                    +\lambda\left(-\int r d\hat{\rho}_{\lambda} + r(M^*) \right)
                 - \mathcal{K}(\hat{\rho}_{\lambda}, \pi)
         - \log\frac{2}{\varepsilon}\Biggr] \geq 0
\Biggr\} 
\leq
 \frac{\varepsilon}{2}
\label{equa:step1}
\end{align}
Using~\eqref{pyta} we have
\begin{align*}
 \int R d\hat{\rho}_{\lambda} - R(M^*) & =
  \int \left\|\Pi_C(XM) -XM^* \right\|_{F,\Pi}^2\hat{\rho}_{\lambda}({\rm d}M)
  \\
 &  \geq   \left\|\int \Pi_C(XM)\hat{\rho}_{\lambda}({\rm d}M) -XM^* \right\|_{F,\Pi}^2
  \\
  & \geq   \left\|\hat{XM}_{\lambda} -XM^* \right\|_{F,\Pi}^2
\end{align*}
where we used Jensen's inequality from the first to the second line, and the definition of $\hat{XM}_{\lambda}$ from the second to the third. Plugging this into our probability bound~\eqref{equa:step1}, and dividing both sides by $\alpha$, we obtain
\begin{equation*}
\mathbb{P}\Biggl\{ \left\|\hat{XM}_{\lambda} -XM^* \right\|_{F,\Pi}^2
\leq 
\frac{ \int r d\hat{\rho}_{\lambda} - r(M^*) +
\frac{1}{\lambda}\left[\mathcal{K}(\hat{\rho}_{\lambda}, \pi) 
+ \log\frac{2}{\varepsilon}\right] } {\frac{\alpha}{\lambda} }
\Biggr\}
 \geq 
 1-\frac{\varepsilon}{2}
\end{equation*}
under the additional condition that $\lambda$ is such that $\alpha>0$, that we will assume from now (note that this is satisfied by $\lambda^*$). Using Lemma~\ref{lemma:dv} we can rewrite this as
%e15 ###
\begin{equation}\label{interm3bis}
\mathbb{P}\Biggl\{ \left\|\hat{XM}_{\lambda} -XM^* \right\|_{F,\Pi}^2
\leq \inf_{\rho \in \mathcal{P}(\mathbb{R}^{m\times p}) } \frac{ \int r
d\rho - r(M^*) +
\frac{1}{\lambda}\left[\mathcal{K}(\rho, \pi)
+ \log\frac{2}{\varepsilon}\right] } {\frac{\alpha}{\lambda} } \Biggr\} \geq 1-\frac{\varepsilon}{2}.
\end{equation}

We stop these derivations for one moment and instead consider now the consequences of the second inequality in Lemma~\ref{lemma:exponential}, that is~\eqref{lemma:exponential:2}. Chernoff's trick and rearranging terms a little, we get
\begin{align*}
 \mathbb{P}\Biggl\{ \forall \rho \in\mathcal{P}(\mathbb{R}^{m\times p}), \int rd\rho - r(M^*)
 \leq 
 \frac{\beta}{\lambda} \left[\int
Rd\rho - R(M^*) \right] + \frac{1}{\lambda}\left[
\mathcal{K}(\rho, \pi) + \log \frac{2}{\varepsilon} \right]
\Biggr\}\geq 1 - \frac{\varepsilon}{2}.
\end{align*}
which we can rewrite as
\begin{multline}
 \mathbb{P}\Biggl\{ \forall \rho \in\mathcal{P}(\mathbb{R}^{m\times p}), \int rd\rho - r(M^*)
\\
 \leq 
 \frac{\beta}{\lambda} \int \|\Pi_C(XM)-X M^* \|_{F,\Pi}^2 \rho({\rm d}M) + \frac{1}{\lambda}\left[
\mathcal{K}(\rho, \pi) + \log \frac{2}{\varepsilon} \right]
\Biggr\}\geq 1 - \frac{\varepsilon}{2}.
\label{interm4}
\end{multline}
\\
Combining (\ref{interm4}) and (\ref{interm3bis}) with a union bound
argument gives the bound
\begin{multline*}
\mathbb{P}\Biggl\{  \left\|\hat{XM}_{\lambda} -XM^* \right\|_{F,\Pi}^2
\leq 
\inf_{\rho \in \mathcal{P}(\mathbb{R}^{m\times p}) } 
\frac{ \beta  \int \|\Pi_C(XM)-X M^* \|_{F,\Pi}^2 
	\rho({\rm d}M) + 2 \left[
\mathcal{K}(\rho, \pi) + \log \frac{2}{\varepsilon} \right] } {
\alpha  } \Biggr\} 
\geq 
1-\varepsilon.
\end{multline*}
It is possible to simplify this bound a little by noting that, for any $(i,j)$, $(X M^*)_{i,j} \in [-C,C]$ implies
that $| (\Pi_C(X M))_{i,j} - (X M^*)_{i,j} | \leq | (X M)_{i,j} - (X M^*)_{i,j} |$ and thus
\begin{align*}
\mathbb{P}\Biggl\{  \left\|\hat{XM}_{\lambda} -XM^* \right\|_{F,\Pi}^2
\leq 
\inf_{\rho \in \mathcal{P}(\mathbb{R}^{m\times p}) } \frac{ \beta  \int \|XM-X M^* \|_{F,\Pi}^2 \rho({\rm d}M) + 2 \left[
\mathcal{K}(\rho, \pi) + \log \frac{2}{\varepsilon} \right] } {
\alpha  } \Biggr\} 
\geq 
1-\varepsilon.
\end{align*}
The end of the proof consists in making the right-hand side in the inequality more explicit. In order to do so, we restrict the infimum bound above to the distributions given by Definition~\ref{dfn:posterior:transla}.
\begin{multline}
\mathbb{P}\Biggl\{  \left\|\hat{XM}_{\lambda} -XM^* \right\|_{F,\Pi}^2
\leq 
\inf_{\bar{M}\in\mathbb{R}^{p\times m} } \frac{ \beta  \int \|XM-X M^* \|_{F,\Pi}^2 \rho_{\bar{M}} ({\rm d}M) + 2 \left[
\mathcal{K}(\rho_{\bar{M}}, \pi) + \log \frac{2}{\varepsilon} \right] } {
\alpha  } \Biggr\} 
\geq 
1-\varepsilon.
\label{PAC_bound}
\end{multline}
We see immediately that Dalalyan's lemma will be extremely useful for that. First, Lemma~\ref{lemma:arnak:2} provides an upper bound on $\mathcal{K}(\rho_{\bar{M}}, \pi)$. Moreover,
\begin{align*}
 & \int  \|XM-X M^* \|_{F,\Pi}^2 \rho_{\bar{M}} ({\rm d}M)
  \\
  & \leq   \int \|X\bar{M}-X M^* - XM \|_{F,\Pi}^2 \pi({\rm d}M)
  \\
  & =   \|X\bar{M}-X M^*  \|_{F,\Pi}^2 -2 \int \sum_{i,j} \Pi_{i,j} (X\bar{M}-X M^*)_{j,i} (XM)_{i,j} \pi({\rm d}M)
 + 
 \int \|XM \|_{F,\Pi}^2 \pi({\rm d}M).
\end{align*}
The second term in the right-hand side is null because $\pi$ is centered, and thus
\begin{align*}
  \int \|XM-X M^* \|_{F,\Pi}^2 \rho_{\bar{M}} ({\rm d}M)
  & \leq   \|X\bar{M}-X M^*  \|_{F,\Pi}^2 + \int \|XM \|_{F,\Pi}^2 \pi({\rm d}M)
  \\
  & \leq   \|X\bar{M}-X M^*  \|_{F,\Pi}^2 +  \int \|X M \|_{F}^2 \pi({\rm d}M)
  \\
  & \leq \|X\bar{M}-X M^*  \|_{F,\Pi}^2 +  \|X \|_{F}^2 \int \|M \|_{F}^2 \pi({\rm d}M)
  \\
  & \leq \|X\bar{M}-X M^*  \|_{F,\Pi}^2 +  \|X \|_{F}^2 pm \tau^2
\end{align*}
where we used elementary properties of the Frobenius norm, and Lemma~\ref{lemma:arnak:1} in the last line. We can now plug this (and Lemma~\ref{lemma:arnak:2}) back into~\eqref{PAC_bound} to get:
\begin{multline*}
\mathbb{P}\Biggl\{  \left\|\hat{XM}_{\lambda} -XM^* \right\|_{F,\Pi}^2
\leq 
\inf_{\bar{M}\in\mathbb{R}^{p\times m} } \Biggl[ \frac{\beta}{\alpha} \|X\bar{M}-X M^*  \|_{F,\Pi}^2 + \frac{\beta}{\alpha} \|X \|_{F}^2 pm \tau^2
\\
+ 
 \frac{1}{\alpha} \left( 4 {\rm rank} (\bar{M}) (m+p+2) \log \left( 1+ \frac{\| \bar{ M } \|_F}{\tau \sqrt{2{\rm rank} (\bar{M})}} \right) + 2 \log \frac{2}{\varepsilon} \right) \Biggr] \Biggr\} 
\geq 
1-\varepsilon.
\end{multline*}
The result is essentially proven, we just explicit the constants.
First, if $\lambda \leq n/(2C_2)$, then $2n( 1 - C_2 \lambda/n ) \geq n$ and thus
$$
\frac{\beta}{\alpha}
= 
\frac{1+ \frac{\lambda C_1 }{2n(1-\frac{ C_2 \lambda}{n})}}{1- \frac{\lambda C_1 }{2n(1-\frac{ C_2 \lambda}{n})}}
 \leq
  \frac{1+ \frac{\lambda C_1 }{n} }{1- \frac{\lambda C_1 }{n} }.
$$
Then, $ \lambda \leq \frac{n \delta}{C_1(1+\delta)} $ leads to
$$
\frac{\beta}{\alpha} \leq (1+\delta).
$$
Note that $\lambda^* = n \min(1/(2C_2), \delta/[C_1(1+\delta)] )$ satisfies these two conditions, so from now $\lambda = \lambda^*$. We also use the following:
\begin{align*}
\frac{1}{\alpha} 
 = 
\frac{1}{\lambda^*\left(1- \frac{\lambda^* C_1 }{2n( 1- C_2 \lambda^* / n )}\right)}
  \leq 
\frac{\beta}{\lambda^* \alpha} 
 \leq 
\frac{(1+\delta)}{n \min(1/(2C_2), \delta/[C_1(1+\delta)] )} 
 \leq 
\frac{C_1 (1+\delta)^2 }{n \delta}.
\end{align*}
So far the bound is: 
\begin{multline*}
\mathbb{P}\Biggl\{  \left\|\hat{XM}_{\lambda^*} -XM^* \right\|_{F,\Pi}^2
\leq 
\inf_{\bar{M}\in\mathbb{R}^{p\times m} } \Biggl[ (1+\delta) \|X\bar{M}-X M^*  \|_{F,\Pi}^2 + (1+\delta) \|X \|_{F}^2 pm \tau^2
\\
+ 
 \frac{C_1 (1+\delta)^2 \left( 4 {\rm rank} (\bar{M}) (m+p+2) \log \left( 1+ \frac{\| \bar{ M } \|_F}{\tau \sqrt{2{\rm rank} (\bar{M})}} \right) + 2 \log \frac{2}{\varepsilon} \right)}{ n\delta } \Biggr] \Biggr\} 
\geq 
1-\varepsilon.
\end{multline*}
In particular, with probability at least $ 1-\varepsilon $, the choice $\tau^2=C_1 (m+p)/(nmp \|X \|_{F}^2)$ gives
\begin{multline*}
 \left\|\hat{XM}_{\lambda^*} -XM^* \right\|_{F,\Pi}^2
\leq 
\inf_{\bar{M}\in\mathbb{R}^{p\times m} } \Biggl[ (1+\delta) \|X\bar{M}-X M^*  \|_{F,\Pi}^2 + \frac{C_1 (1+\delta) (m+p)}{n}
\\
+ 
 \frac{C_1 (1+\delta)^2 \left( 4 {\rm rank} (\bar{M}) (m+p+2) \log \left( 1+\frac{\| X \|_F \| \bar{ M } \|_F}{ \sqrt{C_1}} \sqrt{ \frac{nmp}{(m+p) {\rm rank}(\bar{M})}} \right) + 2 \log \frac{2}{\varepsilon} \right)}{ n\delta } \Biggr].
\end{multline*}
\end{proof}

\subsection{Proof of theorem \ref{thrm_contraction}}

\begin{proof}[\textbf{Proof of Theorem \ref{thrm_contraction}}]
We also start with an application of Lemma~\ref{lemma:exponential}, and focus on~\eqref{lemma:exponential:1}, applied to $\varepsilon:=\varepsilon_n$, that is:
\begin{align*}
 \mathbb{E} \int \exp  \Biggl\{   \alpha    \Bigl( R(M) - R(M^*) \Bigr)
                    +\lambda\Bigl( -r(M) + r(M^*) \Bigr)               - \log \left[\frac{d\hat{\rho}_{\lambda}}{d \pi} (M)  \right]
      - \log\frac{2}{\varepsilon_n}
        \Biggr\}
         \hat{\rho}_{\lambda}(d M)
\leq \frac{\varepsilon_n}{2}.
\end{align*}
Using Chernoff's trick, this gives:
$$
\mathbb{E} \Bigl[ \mathbb{P}_{M\sim \hat{\rho}_{\lambda}} (M\in\mathcal{A}_n) \Bigr]
\geq 1-\frac{\varepsilon_n}{2}
$$
where
$$
\mathcal{A}_n = \left\{M: \alpha    \Bigl( R(M) - R(M^*) \Bigr)
                    +\lambda\Bigl( -r(M) + r(M^*) \Bigr)      \leq      \log \left[\frac{d\hat{\rho}_{\lambda}}{d \pi} (M)  \right]
        + \log\frac{2}{\varepsilon_n} \right\}.
$$
Using the definition of $\hat{\rho}_\lambda $, for $M\in \mathcal{A}_n$ we have
\begin{align*}
\alpha    \Bigl( R(M) - R(M^*) \Bigr)
                    &    \leq  \lambda\Bigl( r(M) - r(M^*) \Bigr)  +       \log \left[\frac{d\hat{\rho}_{\lambda}}{d \pi} (M)  \right]
        + \log\frac{2}{\varepsilon_n}
                    \\
                    & \leq -\log\int\exp\left[-\lambda r(M)\right]\pi({\rm d}M) - \lambda r(M^*)
        + \log\frac{2}{\varepsilon_n}
        \\
        & = \lambda\Bigl( \int r(M) \hat{\rho}_{\lambda}({\rm d}M) - r(M^*) \Bigr)  +    \mathcal{K}(\hat{\rho}_\lambda,\pi)
        + \log\frac{2}{\varepsilon_n}
        \\
        & = \inf_{\rho} \left\{ \lambda\Bigl( \int r(M) \rho({\rm d}M) - r(M^*) \Bigr)  +    \mathcal{K}(\rho,\pi)
        + \log\frac{2}{\varepsilon_n} \right\}.
\end{align*}

\noindent Now, let us define
$$ \mathcal{B}_n = \left\{\forall\rho\text{, }\beta \left(-\int Rd\rho + R(M^*) \right)
+ \lambda \left( \int r d\rho - r(M^*) \right) \leq
\mathcal{K}(\rho, \pi) + \log \frac{2}{\varepsilon_n}\right\}. $$

\noindent Using~\eqref{lemma:exponential:2}, we have that
$$
\mathbb{E} \Bigl[\mathbf{1}_{\mathcal{B}_n} \Bigr]
\geq 1-\frac{\varepsilon_n}{2}.
$$
We will now prove that, if $\lambda$ is such that $\alpha>0$,
$$
\mathbb{E} \Bigl[ \mathbb{P}_{M\sim \hat{\rho}_{\lambda}} (M\in\mathcal{E}_n) \Bigr] \geq \mathbb{E} \Bigl[ \mathbb{P}_{M\sim \hat{\rho}_{\lambda}} (M\in\mathcal{A}_n)\mathbf{1}_{\mathcal{B}_n} \Bigr]
$$
which, together with
\begin{align*}
\mathbb{E} \Bigl[ \mathbb{P}_{M\sim \hat{\rho}_{\lambda}} (M\in\mathcal{A}_n)\mathbf{1}_{\mathcal{B}_n} \Bigr]
& = \mathbb{E} \Bigl[ (1-\mathbb{P}_{M\sim \hat{\rho}_{\lambda}} (M\notin\mathcal{A}_n)) (1-\mathbf{1}_{\mathcal{B}^c_n})\Bigr]
\\
& \geq \mathbb{E} \Bigl[ 1-\mathbb{P}_{M\sim \hat{\rho}_{\lambda}} (M\notin\mathcal{A}_n) - \mathbf{1}_{\mathcal{B}^c_n}
\Bigr]
\\
& \geq 1-\varepsilon_n
\end{align*}
will bring
\begin{equation*}
 \mathbb{E} \Bigl[ \mathbb{P}_{M\sim \hat{\rho}_{\lambda}} (M\in\mathcal{E}_n) \Bigr] \geq 1-\varepsilon_n.
\end{equation*}
In order to do so, assume that we are on the set $\mathcal{B}_n$, and let $M\in\mathcal{A}_n$. Then,
\begin{align*}
\alpha    \Bigl( R(M) - R(M^*) \Bigr)
        & \leq \inf_{\rho} \left\{ \lambda\Bigl( \int r(M) \rho({\rm d}M) - r(M^*) \Bigr)  +    \mathcal{K}(\rho,\pi)
        + \log\frac{2}{\varepsilon_n} \right\}
        \\
        & \leq \inf_{\rho} \left\{ \beta \Bigl( \int R(M) \rho({\rm d}M) - R(M^*) \Bigr)  +   2 \mathcal{K}(\rho,\pi)
        + 2 \log\frac{2}{\varepsilon_n} \right\}
\end{align*}
that is,
$$
R(M) - R(M^*) \leq \inf_{\rho \in \mathcal{P}(\mathbb{R}^{m\times p}) } \frac{ \beta \left[\int Rd\rho -
R(M^*) \right] + 2 \left[
\mathcal{K}(\rho, \pi) + \log \frac{2}{\varepsilon} \right] } {
\alpha  }
$$
or, rewriting it in terms of norms,
$$
 \left\|\Pi_C(XM) -XM^* \right\|_{F,\Pi}^2
\\
\leq 
\inf_{\bar{M}\in\mathbb{R}^{p\times m} } \frac{ \beta  \int \|XM-X M^* \|_{F,\Pi}^2 \rho_{\bar{M}} ({\rm d}M) + 2 \left[
\mathcal{K}(\rho_{\bar{M}}, \pi) + \log \frac{2}{\varepsilon} \right] } {
\alpha  }.
$$
We upper-bound the right-hand side exactly as in the proof of Theorem~\ref{thrm_contraction}, this gives
$M\in\mathcal{E}_n$.

\end{proof}

\newpage
\section{Appendix: Additional simulations}
\label{ap_simu}

%tb_model 1 with replacment
\begin{table}[!h]
	\caption{Simulation results on simulated data in Setting I ($ \rho_X = 0 $) with replacement ($ \ell = 100,  p = 8, m = 12 $ and $r = 2 $) for different methods, with their standard error in parentheses.  (Est: average of estimation error; Pred: average of prediction error).}
	\centering
	\small
	\makebox[\textwidth]{
		\begin{tabular}{ | p{11mm}|c | c|c || c|c|c | } 
			\toprule
			& \multicolumn{3}{ c|| }{observed 80\% without replacement} 
			& \multicolumn{3}{ c | }{observed 80\% with replacement} 
			\\
			Errors    	& LMC & MALA & mRRR & LMC & MALA & mRRR
			\\ 
			\hline
	Est 	& 0.165 (0.026) & 0.163 (0.026) & 0.072 (0.038)
			& 0.190 (0.030) & 0.188 (0.030) & NA
			\\ 
	Pred 	& 2.234 (0.236) & 2.233 (0.234) & 2.099 (0.216)
			& 0.222 (0.044) & 0.220 (0.043) & NA 
			\\  
			\midrule
			& \multicolumn{3}{ c || }{observed 50\% without replacement} 
			& \multicolumn{3}{ c | }{observed 50\% with replacement} 
			\\ 
			\hline
	Est 	& 0.293 (0.052) & 0.288 (0.051) & 0.315 (0.278)
			& 0.341 (0.071) & 0.336 (0.071) & NA
			\\
	Pred  	& 2.322 (0.160) & 2.313 (0.157) & 2.315 (0.338)
			& 0.397 (0.097) & 0.391 (0.096) & NA
			\\  
			\midrule
			& \multicolumn{3}{ c || }{observed 20\% without replacement} 
			& \multicolumn{3}{ c | }{observed 20\% with replacement} 
			\\ 
			\hline
	Est 	& 14.41 (26.59) & 1.689 (0.507) & 1.707 (0.614)
			& 8.948 (9.350) & 2.280 (0.927) & NA
			\\
	Pred  	& 19.79 (32.99) & 3.966 (0.632) & 3.718 (0.653)
			& 10.76 (11.39) & 2.651 (1.125) & NA
			\\ \bottomrule
		\end{tabular}
	}
	\label{tb_m_1_wr}
\end{table}

%tb_model 3 with replacement
\begin{table}[!h]
	\caption{Simulation results on simulated data in Setting III, approximate low-rank model, with replacement for different methods, with their standard error in parentheses. (Est: average of estimation error; Pred: average of prediction error).}
	\centering
	\small
	\makebox[\textwidth]{
		\begin{tabular}{ | p{11mm} | c| c| c || c| c| c | } 
			\toprule
			& \multicolumn{3}{ c|| }{observed 80\% without replacement} 
			& \multicolumn{3}{ c | }{observed 80\% with replacement} 
			\\
			Errors    	& LMC & MALA & mRRR 
			& LMC & MALA & mRRR
			\\ \midrule
	Est 	& 0.157 (0.026) & 0.156 (0.025) & 0.313 (0.671)
			& 0.185 (0.026) & 0.183 (0.027) & NA
			\\ 
	Pred  	& 2.184 (0.220) & 2.182 (0.218) & 2.277 (0.552)
			& 0.215 (0.038) & 0.214 (0.039) & NA
		
			\\ 
			\midrule
			& \multicolumn{3}{ c || }{observed 50\% without replacement} 
			& \multicolumn{3}{ c | }{observed 50\% with replacement} 
			\\ \midrule
	Est 	& 0.286 (0.050) & 0.281 (0.049) & 0.563 (0.660)
			& 0.340 (0.062) & 0.337 (0.061) & NA
			\\
	Pred  	& 2.304 (0.182) & 2.298 (0.103) & 2.565 (0.733)
			& 0.397 (0.084) & 0.392 (0.082) & NA
			\\ 
			\midrule
			& \multicolumn{3}{ c || }{observed 20\% without replacement} 
			& \multicolumn{3}{ c | }{observed 20\% with replacement} 
			\\ \midrule
	Est 	& 7.766 (15.18) & 1.708 (0.683) & 2.985 (1.971)
			& 7.049 (6.238) & 2.207 (0.803) & NA
			\\
	Pred 	& 11.51 (18.81) & 3.985 (0.907) & 5.097 (2.069)
			& 8.473 (7.611) & 2.569 (0.980) & NA
			\\ 
			\bottomrule
		\end{tabular}
	}
	\label{tb_m_3_wr}
\end{table}

\clearpage
\begin{footnotesize}
\bibliographystyle{abbrv}

\end{footnotesize}

\end{document}